\pdfoutput=1
\documentclass[11pt,drafunctiontcls,onecolumn]{IEEEtran}
\usepackage{etex}
\usepackage{amssymb}
\usepackage{times}
\usepackage{bm}
\usepackage{subcaption}
\usepackage{amsmath}
\usepackage{graphicx}
\usepackage{amssymb}
\usepackage{amstext}
\usepackage{latexsym}
\usepackage{diagbox}
\usepackage{color,soul}
\usepackage{placeins}
\usepackage{ifthen}
\usepackage{multirow}
\usepackage{verbatim}
\usepackage{array,tabularx}
\usepackage{accents}
 \usepackage{cite}
 \usepackage{arydshln}
\usepackage{url,amssymb,epsfig,times,mathrsfs,algorithm,algorithmic,latexsym,xspace,fancyhdr,wrapfig, xcolor,multirow,amsthm}
\usepackage{caption}
\usepackage{soul}
\usepackage[inline]{enumitem}
\usepackage{arydshln}
\usepackage{enumitem}
\usepackage[makeroom]{cancel}
\usepackage{mathtools}
\usepackage{dblfloatfix}
\usepackage[utf8]{inputenc}

\usepackage{tikz}
\usetikzlibrary{shapes.geometric}
\usetikzlibrary{shapes,snakes,patterns}
\usetikzlibrary{arrows,positioning,automata,calc}
\usetikzlibrary{plotmarks}
\tikzset{
    int/.style={
           rectangle,
           rounded corners,
           draw=black, thin, fill=black!20,
           minimum height=2em,
           inner sep=2pt,
           text centered,
           },
}

\newtheorem{theorem}{Theorem}
\newtheorem{lemma}{Lemma}
\newtheorem{claim}{Claim}
\newtheorem{proposition}{Proposition}
\newtheorem{corollary}{Corollary}
\newtheorem{remark}{Remark}

\IEEEoverridecommandlockouts

\newcommand{\Mod}[1]{\ (\mathrm{mod}\ #1)}
\pdfoutput=1

\begin{document}
\pdfoutput=1
\allowdisplaybreaks
\newlength\figureheight
\newlength\figurewidth

\title{ List Decoding of Deletions Using \\ Guess \& Check Codes \vspace{-0.2cm} \thanks{This work was supported  in parts by NSF Grant CCF 15-26875.}}
\author{
\IEEEauthorblockN{Serge Kas Hanna,  Salim El Rouayheb \\ ECE Department, Rutgers University\\ serge.k.hanna@rutgers.edu, salim.elrouayheb@rutgers.edu
\vspace{-0.1cm}
}
}
\maketitle
\begin{abstract} 
Guess \& Check (GC) codes are systematic binary codes that can correct multiple deletions, with high probability. GC codes have logarithmic redundancy in the length of the message $k$, and the encoding and decoding algorithms of these codes are deterministic and run in polynomial time for a constant number of deletions $\delta$. The {\em unique} decoding properties of GC codes were examined in a previous work by the authors. In this paper, we investigate the {\em list decoding} performance of these codes. Namely, we study the {\em average} size and the {\em maximum} size of the list obtained by a GC decoder for a constant number of deletions~$\delta$. The theoretical results show that:~(i)~the average size of the list approaches $1$ as $k$ grows; and~(ii)~there exists an infinite sequence of GC codes indexed by $k$, whose maximum list size in upper bounded by a constant that is independent of $k$. We also provide numerical simulations on the list decoding performance of GC codes for multiple values of~$k$ and~$\delta$.
\end{abstract}

\section{Introduction}
\label{sec:1}
\label{sec:1}
Codes that correct deletions have several applications such as file synchronization and DNA-based storage. In remote file synchronization for instance, the goal is to synchronize two edited versions of the same file with minimum number of communicated bits. In general, the file edits are a result of a series of deletions and insertions. One way to achieve synchronization is by using systematic codes that can efficiently correct deletions~\cite{GC,V15}.

The study of deletion correcting codes goes back to the 1960s~\cite{G61,VT65,L66}. In 1965, Varshamov and Tenengolts constructed the VT codes for correcting asymmetric errors over the Z-channel~\cite{VT65}. In 1966, Levenshtein showed that VT codes are capable of correcting a single deletion with zero-error~\cite{L66}. Also in~\cite{L66}, Levenshtein derived fundamental bounds on the redundancy needed to correct $\delta$ deletions. The bounds showed that the number of redundant bits needed to correct $\delta$ deletions in a codeword of length $n$ bits is logarithmic in $n$, namely $c\delta \log n$~bits for some constant $c>0$. The fundamental bounds on the redundancy were later generalized and improved by Cullina and Kiyavash~\cite{N14}. The redundancy of VT codes is asymptotically optimal for correcting a single deletion. Finding VT-like codes for multiple deletions has been an open problem for several decades. The literature on multiple deletion correcting codes has mostly focused on constructing codes that can decode multiple deletions with zero-error~\cite{H02,B16,Gab17,S99,G14}. There has also been multiple recent works which study codes that can correct multiple deletions with low probability of error~\cite{GC, Ab17, T16, Ob}.  

In this work, we are interested in the problem of designing efficient codes for list decoding of deletions. A list decoder returns a list of candidate strings which is guaranteed to contain the actual codeword. The idea of list decoding was first introduced in the 1960s by Elias~\cite{E57} and Wozencraft~\cite{W58}. The main goal when studying list decoders is to find explicit codes that can return a small list of candidate strings in polynomial time\footnote{Polynomial time in terms of the length of the codeword $n$.}. The size of the list gives a lower bound on the time complexity of the list decoder. For instance, if the size of the list is superpolynomial, then polynomial time list decoding cannot be achieved. List decoding has been studied for various classes of error correcting codes, such as Reed-Solomon codes~\cite{RSlist}, Reed-Muller codes~\cite{RMlist}, and Polar codes~\cite{Plist}. However, the problem of finding list decoders for deletions has not received much attention in the literature. In~\cite{G14}, Guruswami and Wang proved the {\em existence} of codes that can list-decode a constant fraction of deletions given by $n(\frac{1}{2}-\epsilon)$, where $0<\epsilon <\frac{1}{2}$ and $n$ is the length of the codeword. In the regime considered in~\cite{G14}, the codes have low rate of the order of~$\epsilon^3$. Recently in~\cite{A17}, Wachter-Zeh derived an upper bound on the list size for decoding deletions and insertions. An explicit  list decoding algorithm that is based on VT codes was also proposed in~\cite{A17}.

In this paper, we focus on the case where the codeword $\mathbf{x}$ is affected by a {\em constant} number of deletions $\delta$, resulting in a string $\mathbf{y}$ of length $n-\delta$. A simple list decoder in this case is one that returns all binary strings that have a Levenshtein distance $\delta$ from $\mathbf{y}$, i.e., all supersequences of $\mathbf{y}$ of length $n$. This list decoder does not require any redundancy, and its resulting list size is exactly $\sum_{i=0}^{\delta} \binom{n}{i}$ \cite{Orlitsky}. Hence, for a constant number of deletions $\delta$, the list size is $\mathcal{O}(n^{\delta})$ (polynomial function of $n$ of degree $\delta$). In~\cite{A17}, this list size was reduced to $\mathcal{O}(n^{\delta-1})$  by using VT codes. The idea in~\cite{A17} is to first generate all binary strings that have a Levenshtein distance $\delta-1$ from~$\mathbf{y}$, and then decode these strings using VT codes. Note that this reduction in the list size comes at the expense of adding a logarithmic redundancy that is introduced by VT codes. Our main contribution in this paper is showing that the Guess \& Check~(GC) codes which we presented in~\cite{GC}, can achieve a list size that is upper bounded by constant, i.e., $\mathcal{O}(1)$, with logarithmic redundancy and in polynomial time.

GC codes are explicit systematic codes that have logarithmic redundancy and can correct a constant number of deletions $\delta$ in polynomial time. Initially, the unique decoding performance of GC codes was studied in~\cite{GC}. In the unique decoding setting, a decoding failure is declared if the decoding results in a list that contains more than one candidate string. The study in~\cite{GC} showed that the probability of decoding failure of GC codes vanishes asymptotically in the length of the message $k$ (or equivalently the length of the codeword $n$). In this work, we quantify the value of the list size obtained by the GC decoder by studying the {\em average} and the {\em maximum} size of the list. Through our theoretical and simulation results, we show that for a constant number of deletions $\delta$, GC codes can return a small list of candidate strings in polynomial time. Namely, our contributions are the following.

{\em Theoretical results:} Our theoretical results show that:~(i)~the average size of the list approaches~$1$~asymptotically in~$k$ for a uniform iid message and any deletion pattern; and~(ii)~there exists an infinite sequence of GC codes indexed by $k$ whose maximum list size is upper bounded by a constant that is independent of~$k$, for any message and any deletion pattern. These results demonstrate that in the average case, lists of size strictly greater than one occur with low probability. Furthermore, in the worst case, the list size is very small~(constant), and hence the performance of GC codes is very close to codes that can uniquely decode multiple deletions with zero-error.

{\em Simulation results:} We provide numerical simulations on the list decoding performance of GC codes for values of $k$ up to $1024$ bits, and values of $\delta$ up to $3$ deletions. The average list size recorded in these simulations is very close to 1 (less than 2). Whereas, the maximum list size detected within the performed simulations is $3$. 

{\em Comparison to~\cite{A17}:} Our theoretical results improve on the list decoder presented in~\cite{A17}, whose maximum list size is theoretically $\mathcal{O}(n^{\delta-1})$, i.e., upper bounded by a function that grows polynomially in length of the codeword $n$ for a constant number of deletions~$\delta$. Furthermore, in Section~\ref{sec:5}, we provide a numerical comparison to~\cite{A17} which shows that the maximum list size of list decoder in~\cite{A17} grows with $n$ and is much larger than that of GC codes.

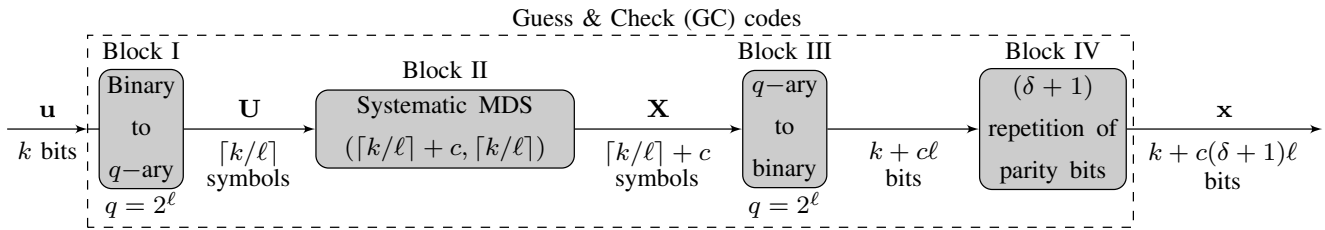
\begin{figure*}
\centering
\resizebox{1\textwidth}{!}{
\begin{tikzpicture}[node distance=2.5cm,auto,>=latex']
\draw (-0.6,0) -- (-0.45,0);
    \node [int] (c) [text width=0.8cm,align=center]{\scriptsize Binary to $q-$ary};
    \node (b) [left of=c,node distance=1.5cm, coordinate] {a};
    \node [int] (z) [right of=c, node distance=3.4cm,text width=2.75cm,align=center] {\scriptsize Systematic MDS $\left(\left \lceil k/\ell \right \rceil+c,\left \lceil k/\ell \right \rceil \right)$};
    \node [int] (y) [right of=z, node distance=3.8cm,text width=0.8cm,align=center] {\scriptsize $q-$ary to binary};
    \node [int] (y1) [right of=y, node distance=3cm, text width=1.5cm,align=center] {\scriptsize {$(\delta+1)$ repetition of parity bits}};
    \node [coordinate] (end) [right of=c, node distance=2cm]{};
    \path[->] (b) edge node {\scriptsize $\mathbf{u}$} node[below] {\scriptsize $k$ bits} (-0.6,0);
    \draw[->] (c) edge node {\scriptsize $\mathbf{U}$} (z) ;
     \node(h1) [right of=c,node distance=1.2cm] [below] {\scriptsize $\left \lceil k/\ell \right \rceil$};
    \node(h2) [below of=h1,node distance=0.3cm] {\scriptsize symbols};
    \node(q1) [below of=c,node distance=0.85cm] {\scriptsize $q=2^{\ell}$};
    \node(q2) [below of=y,node distance=0.85cm] {\scriptsize $q=2^{\ell}$};
    \path[->] (z) edge node {\scriptsize $\mathbf{X}$} (y); 
    \node(i1) [right of=z,node distance=2.37cm] [below] {\scriptsize $\left \lceil k/\ell \right \rceil+c$};
   \node(name)[above of=i1,node distance=1.25cm] [above] {\scriptsize Guess \& Check (GC) codes};
    \node(i2) [below of=i1,node distance=0.3cm] {\scriptsize symbols};
    \node(e1) [right of=y1,node distance=3.15cm] {};
    \path[->] (y1) edge node {\scriptsize $\mathbf{x}$} (e1);
    \path[->] (y) edge node {} (y1);
    \node(t1) [right of=y,node distance=1.33cm] [below] {\scriptsize $k+c\ell$};
    \node(t2) [below of=t1,node distance=0.3cm] {\scriptsize bits};
    \node(tt1) [right of=y1,node distance=1.9cm] [below] {\scriptsize $k+c(\delta +1)\ell$};
    \node(tt2) [below of=tt1,node distance=0.3cm] {\scriptsize bits}; 
    \node(b1) [above of=c,node distance=0.88cm] {\scriptsize Block I};
    \node(b2) [above of=z,node distance=0.665cm] {\scriptsize Block II};
    \node(b3) [above of=y,node distance=0.88cm] {\scriptsize Block III};
    \node(b4) [above of=y1,node distance=0.88cm] {\scriptsize Block IV};
    
    \draw[dashed] (-0.6,-1.1) rectangle (11.1,1.05);
\end{tikzpicture} }
\captionsetup{font=footnotesize}
\caption{Encoding block diagram of Guess \& Check (GC) codes for $\delta$ deletions. Block I: The binary message of length $k$ bits is chunked into adjacent blocks of length $\ell$ bits each, and each block is mapped to its corresponding symbol in $GF(q)$ where $q=2^{\ell}$. Block II: The resulting string is coded using a systematic $\left(\left \lceil k/\ell \right \rceil+c,\left \lceil k/\ell \right \rceil \right)$ $q-$ary MDS code where $c>\delta$ is the number of parity symbols. Block~III:~The symbols in $GF(q)$ are mapped to their binary representations. Block IV: Only the parity bits are coded using a $(\delta+1)$ repetition code.}
\label{fig:1}
\vspace{-0.5cm}
\end{figure*}

\section{Preliminaries} 
\label{sec:2}
In this section, we present an overview of Guess \& Check (GC) codes~\cite{GC} and state some of the previous results on these codes which will be helpful in subsequent sections of this paper. We also introduce the necessary notations used throughout the paper.
\subsection{Guess \& Check (GC) Codes for List Decoding}
GC codes were presented in~\cite{GC} as explicit binary codes that can correct multiple deletions, with high probability. These codes can also be used for the list decoding of deletions as we describe in the encoding and decoding steps below.

Let $\mathbf{u}\in \mathbb{F}_2^{k}$ be a message of length $k$~bits. Let $\mathbf{x}\in \mathbb{F}_2^n$ be its corresponding codeword of length $n$ bits. Let $\delta$ be a constant representing the number of deletions. Consider a deletion pattern $\mathbf{d}=(d_1,d_2,\ldots,d_{\delta})$ representing the positions of $\delta$ deletions, where $d_i\in \{1,2,\ldots,n\}$, for \mbox{$i=1,\ldots,\delta$}. 

\begin{enumerate}[leftmargin=*]
\item {\em Encoding:} The encoding block diagram is illustrated in~Fig.~\ref{fig:1}. Encoding is done based on the following steps. (i)~The message $\mathbf{u}$ of length $k$ bits is chunked into $\left \lceil k/\ell \right \rceil$ adjacent blocks of length $\ell$ bits each, and each block is then mapped to its corresponding symbol in $GF(q)$, where $q=2^{\ell}$. (ii)~The resulting $q-$ary string is encoded using a systematic $\left(\left \lceil k/\ell \right \rceil+c,\left \lceil k/\ell \right \rceil \right)$ MDS code where $c>\delta$ is a constant representing the number of parity symbols. (iii)~The $q-$ary symbols are mapped backed to their binary representations. (iv)~Only the parity bits are encoded using a $(\delta+1)$ repetition code. This encoding procedure results in the codeword $\mathbf{x}$ of length $n$ bits, where $n=k+c(\delta+1)\ell$.
\item {\em Decoding:} Suppose that the codeword $\mathbf{x}$ is affected by a deletion pattern $\mathbf{d}=(d_1,d_2,\ldots,d_{\delta})$, resulting in a string $\mathbf{y}$ of length $n-\delta$ that is received by the decoder. The received string $\mathbf{y}$ is decoded based on the following steps. (i)~The decoder recovers the parity bits which are protected by a $(\delta+1)$ repetition code. (ii)~The decoder makes $t$ guesses, where each guess corresponds to a specific distribution of the $\delta$ deletions among the $\left \lceil k/\ell \right \rceil$ blocks. The total number of guesses is equal to the total number of possible deletion patterns given by
\begin{equation}
\label{eqt}
t=\binom{\left \lceil k/\ell \right \rceil + \delta -1}{\delta} =\mathcal{O} \left( \frac{k^{\delta}}{\ell^{\delta}} \right).
\end{equation} 
(iii)~For each guess, the decoder specifies the block boundaries based on its assumption on the locations of the $\delta$ deletions. Then, it treats the blocks that are affected by deletions as hypothetical erasures and decodes these erasures using the first $\delta$ MDS parity symbols. This guessing phase results in an {\em initial} list of at most\footnote{Depending on the runs of bits within the received string $\mathbf{y}$, different guesses may lead to the same decoded string.}~$t$ decoded strings. (iv)~The decoder then checks whether each decoded string in the initial list is a valid guess or not, and removes the invalid ones. A guess is considered is to be valid if the decoded string is consistent with the remaining $c-\delta$ parities; and its Levenshtein distance from the received string $\mathbf{y}$ is exactly~$\delta$. At the end of this checking phase, the GC decoder is left with a smaller {\em final} list of candidate strings.
\end{enumerate}
\begin{proposition}
The final list returned by the GC decoder is guaranteed to contain the actual codeword $\mathbf{x}$, and all the strings in this list have a Levenshtein distance $\delta$ from the received string $\mathbf{y}$.
\end{proposition}
\begin{proof}
The decoder goes over all possible deletion patterns, so the actual deletion pattern is guaranteed to be considered in one of the $t$ guesses. Furthermore, the parities are recovered with zero-error since they are protected by a $(\delta+1)$ repetition code. Therefore, the decoding of $\mathbf{y}$ for the correct guess will result in the actual codeword~$\mathbf{x}$. Also, the fact that all the strings in this list have a Levenshtein distance $\delta$ from $\mathbf{y}$, follows directly from part (iv) of the decoding steps.
\end{proof}

Since the size of the initial list is at most $t$, and $t$ is upper bounded by a polynomial function of $k$ given by~\eqref{eqt}, then the initial list size is at most {\em polynomial} in $k$. In this paper, we are interested in studying the size of the final list obtained by the GC decoder. The list size is a deterministic function of the codeword $\mathbf{x}$ and the deletion pattern~$\mathbf{d}$, and hence can be represented by $L(\mathbf{x},\mathbf{d})$. Since the codeword~$\mathbf{x}$ is a deterministic function of the message $\mathbf{u}$, an equivalent definition of the list size is $L(\mathbf{u},\mathbf{d})$. Throughout the paper, we drop the $\mathbf{d}$ argument and use $L(\mathbf{u})$
to refer to the maximum list size over all possible deletion patterns for a message $\mathbf{u}$ of length $k$ bits, i.e., 
\begin{equation}
\label{def1}
L(\mathbf{u})\triangleq \max_{\mathbf{d}}L(\mathbf{u},\mathbf{d}).
\end{equation}
Based on the decoding steps of GC codes, we have \mbox{$L(\mathbf{u})\in \{1, \ldots, t\}$}. To quantify the size of the final list we define the following quantities:
\begin{enumerate}[leftmargin=*]
\item The {\em average} value of the list size, defined by
\begin{equation}
\label{def2}
L_{av}\triangleq \mathbb{E}(L(\mathbf{u})) = \sum_{l=1}^{t} l \cdot Pr(L(\mathbf{u})=l),
\end{equation}
for a uniform iid message $\mathbf{u}$ of length $k$ bits.
\item The {\em maximum} value of the list size, defined by
\begin{equation}
\label{lmax}
L_{max}\triangleq \max_{\mathbf{u}}L(\mathbf{u}),
\end{equation}
for any message $\mathbf{u}$ of length $k$ bits.
\end{enumerate} 
The definition of $L(\mathbf{u})$ in~\eqref{def1} implies that the average $L_{av}$ defined in~\eqref{def2}, and the maximum $L_{max}$ defined in~\eqref{lmax}, are maximized over all possible deletion patterns. 

\subsection{Previous Results on Unique Decoding}\
GC codes are systematic codes, and it follows from the encoding block diagram in Fig.~\ref{fig:1} that their redundancy is $n-k=c(\delta+1)\ell$. The encoding and decoding algorithms of GC codes are deterministic and run in polynomial time for a constant number of deletions $\delta$~\cite{GC}. In the unique decoding setting, successful decoding is declared if only one guess is valid. If two or more guesses are valid, the decoder declares a decoding failure. The results in~\cite{GC} show that an upper bound on the probability of decoding failure for a uniform iid message~$\mathbf{u}$, and any deletion pattern $\mathbf{d}$, is given by
\begin{equation}
Pr(F)=\mathcal{O}\left(\frac{k^{\delta}}{\ell^{\delta}2^{\ell(c-\delta)}}\right). \label{eqprf}
\end{equation}
If $\ell=\Omega(\log k)$ and $c\geq 2\delta$, then this probability of decoding failure goes to zero as $k$ goes to infinity.

\subsection{Notations}
We summarize the notations used in this paper in Table~I.
\begin{table}[h]
\centering
\setlength\extrarowheight{1.6pt}
 \begin{tabular}{|c|l|c|l|}
\hline
Variable & Description & Variable & Description \\ \hline
$\mathbf{u}$ & message & $c$ & number of MDS parity symbols \\ \hline
$k$ & length of the message in bits & $\ell$ & chunking length used for encoding \\ \hline
$\mathbf{x}$ & codeword & $q$ & field size given by $2^{\ell}$ \\ \hline
$n$ & length of codeword in bits & $t$ & total number of guesses given in~\eqref{eqt} \\ \hline
$\mathbf{d}$ & deletion pattern & $l$ & realization of the list size \\ \hline
$\delta$ & number of deletions & $L_{av}$ & average list size defined in~\eqref{def2} \\ \hline
$d_i$ & position of the $i^{th}$ deletion & $L_{max}$ & maximum list size defined in~\eqref{lmax} \\ \hline
\end{tabular}
\captionsetup{font=footnotesize}
\caption{\footnotesize{Summary of the notations used in the paper.}}
\vspace{-0.3cm}
\label{t}
\end{table}

\section{Main Results}
\label{sec:3}
In this section, we present our main results in this paper. The proofs of these results are given in Section~\ref{sec:4}. Recall that the number of deletions $\delta$ and the number of parity symbols $c$ are constants, i.e., independent of $k$. 
\subsection{Results}
Theorem~\ref{thm:1} gives an upper bound on the average list size $L_{av}$ defined in~\eqref{def2}, in terms of the GC code parameters summarized in Table~\ref{t}.
\begin{theorem}[Average list size]
\label{thm:1} 
For a uniform iid message of length $k$ bits, and any deletion pattern $(d_1,d_2,\ldots,d_{\delta})$, the average list size $L_{av}$ obtained by the Guess \& Check (GC) decoder satisfies
\begin{equation*}
1\leq L_{av} \leq 1 +  \mathcal{O}\left(\frac{(k/\ell)^{2\delta}}{2^{\ell(c-\delta)}}\right).
\end{equation*}
\end{theorem}
The next result follows from Theorem~\ref{thm:1} and shows that for appropriate choices of the GC code parameters, $L_{av}$ approaches one asymptotically in the length of the message~$k$.
\begin{corollary}
\label{c:1} 
For choices of the GC code parameters that satisfy $\ell=\Omega(\log k)$ and \mbox{$c\geq 3\delta$}, the average list size satisifies 
\begin{equation*}
\lim_{k\to +\infty}L_{av}=1.
\end{equation*}
\end{corollary}
The next theorem presents an upper bound on the maximum list size~$L_{max}$ defined in~\eqref{lmax}, which is the main quantity of interest in the list decoding literature. Informally, the theorem states that for appropriate choices of the code parameters, the maximum list size $L_{max}$ is upper bounded by a constant,
\begin{theorem}[Maximum list size]
\label{thm:2} 
Let $\ell=\Omega(\log k)$ and \mbox{$c\geq 2\delta$}. Consider a sufficiently large message length $k_1$. There exists an infinite sequence of Guess \& Check (GC) codes indexed by the message lengths \mbox{$k_1<k_2<k_3<\ldots$}, whose maximum list size $L_{max}$, for any message of length $k\in \{k_1,k_2,k_3,\ldots\}$, and any deletion pattern $(d_1,\ldots,d_{\delta})$, is upper bounded by a constant that is independent of~$k$.
\end{theorem}
Theorem~\ref{thm:2} says that there exists an infinite sequence of GC codes whose $L_{max}$ is upper bounded by a constant. The restriction to a sequence of codes is a limitation of the proof technique. We conjecture that this result is true for all GC codes with arbitrary $k$. In Section~\ref{sec:5}, we provide numerical simulations on the maximum list size, where we gradually increase the message length from $k=32$ to $k=1024$ for multiple values of $\delta$. In the obtained empirical results, the value of the maximum list size does not increase with $k$.

\subsection{Discussion for the case of $\ell=\log k$} 
Recall that the redundancy of GC codes is $n-k=c(\delta+1)\ell$. Let $\ell=\log k$ be the chunking length used for encoding. In this case, the redundancy is $c(\delta+1)\log k$, i.e., logarithmic in $k$. It follows from Theorems~\ref{thm:1}~and~\ref{thm:2} that a logarithmic redundancy is sufficient for GC codes so that~(i)~\mbox{$\lim_{k\to \infty} L_{av} = 1$}; and (ii)~\mbox{$L_{max}=\mathcal{O}(1)$} for an infinite sequence of GC codes. It is easy to verify that a logarithmic redundancy corresponds to a code rate $R=\frac{k}{n}$ that is asymptotically optimal in $n$ (rate approaches one as $n$ goes to infinity). Therefore, GC codes can achieve the list decoding properties given in Theorems~\ref{thm:1}~and~\ref{thm:2} with a logarithmic redundancy and an asymptotically optimal code rate.

\section{Simulations}
In this section we present simulation results on the average and maximum list size obtained by GC codes. We also compare the list decoding performance of GC codes to that of the codes presented in~\cite{A17}.
\label{sec:5}
\subsection{Simulation results for GC codes}
We performed numerical simulations on the average list size $L_{av}$ and the maximum list size $L_{max}$ for $k~=32,64,128,256,512$ and $1024$ bits, and for $\delta=1, 2$ and~$3$ deletions. The empirical results are shown in~Table~II.
\begin{table}[htb]
\centering
\setlength\extrarowheight{1.2pt}
 \begin{tabular}{|c|c|c|c|c|c|c|c|c|}
\hline
\multirow{2}{*}{Config.} & \multicolumn{6}{c|}{$\delta$} \\ \cline{2-7}
& \multicolumn{2}{c|}{$1$}  & \multicolumn{2}{c|}{$2$} &  \multicolumn{2}{c|}{$3$} \\ \hline
$k$ & $L_{av}$ & $L_{max}$ & $L_{av}$ & $L_{max}$ &  $L_{av}$ & $L_{max}$ \\ \hline
32 & $1.0183$ & $2$ & $1.0151$ & $3$ & $1.0061$ &$3$ \\ \hline
64 & $1.0110$ & $2$ & $1.0061$ & $3$ & $1.0043$ &$3$ \\ \hline
128 & $1.0069$ & $2$ & $1.0029$ & $2$ & $1.0020$ &$2$ \\ \hline
256 & $1.0035$ & $2$ & $1.0020$ & $2$ & $1.0007$ &$2$ \\ \hline
512 & $1.0021$ & $2$ & $1.0010$ &$2$ & $1.0005$ &$2$  \\ \hline
1024 & $1.0007$ & $2$ & $1.0005$ &$2$ & $1.0002$ &$2$  \\ \hline
\end{tabular}
\captionsetup{font=footnotesize}
\caption{\small{The table shows the average list size $L_{av}$ and the maximum list size $L_{max}$ obtained by the GC decoder for different message lengths $k$ and different number of deletions $\delta$. The results shown are for $c=\delta+1$ and $\ell=\log k$. The results of $L_{av}$ and $L_{max}$ were recorded over $10000$ runs of simulations. In each run, a message $\mathbf{u}$ chosen uniformly at random is encoded into the codeword $\mathbf{x}$. $\delta$ bits are then deleted from $\mathbf{x}$ based on a uniformly distributed deletion pattern $\mathbf{d}$, and the resulting string is decoded.
}}
\label{ttt}
\end{table}
\FloatBarrier
The results show that: (i)~the average list size $L_{av}$ is very close to one and its value approaches one further as $k$ increases; and (ii)~the maximum list size $L_{max}$ recorded is $3$ and $L_{max}$ does not increase with $k$, for $k=32,64,\ldots, 1024$. 
Note that the redundancy used for these simulations is \mbox{$n-k=c(\delta+1)\log k$} with $c=\delta+1$. This redundancy is much smaller than the one suggested by Theorems~\ref{thm:1}~and~\ref{thm:2}. This is due to the fact that GC codes perform better than what the theoretical bounds indicate, which was discussed in~\cite{GC}.

\subsection{Comparison with the codes in~\cite{A17}}
 In~\cite{A17}, a list decoder of $\delta$ deletions was presented that is based on VT codes~\cite{VT65}. Recall that VT codes can uniquely decode a single deletion. Consider a codeword that is affected by $\delta$ deletions, resulting in a string $\mathbf{y}$ of length $n-\delta$ bits.  The main idea in~\cite{A17} is to first generate all the supersequences of $\mathbf{y}$ of length $n-1$; and then apply the VT decoder on each supersequence. The decoding results in a list whose maximum size is theoretically $\mathcal{O}(n^{\delta-1})$. Note that this size increases polynomially in $n$ for a constant $\delta$. We simulated the maximum list size of the list decoder in~\cite{A17}, and compared it to that of GC codes for $\delta=2$. The obtained results are shown in Fig.~\ref{fig::2}. The comparison shows that the maximum list size in~\cite{A17} is larger and increases with the message length $k$. Note that the two compared codes have the same order of redundancy (logarithmic in $n$) and the same order of decoding complexity (polynomial in $n$ for constant $\delta$).
\begin{figure}[h!]
\hspace{-0.4cm}
\centering
\includegraphics[width=0.45\textwidth]{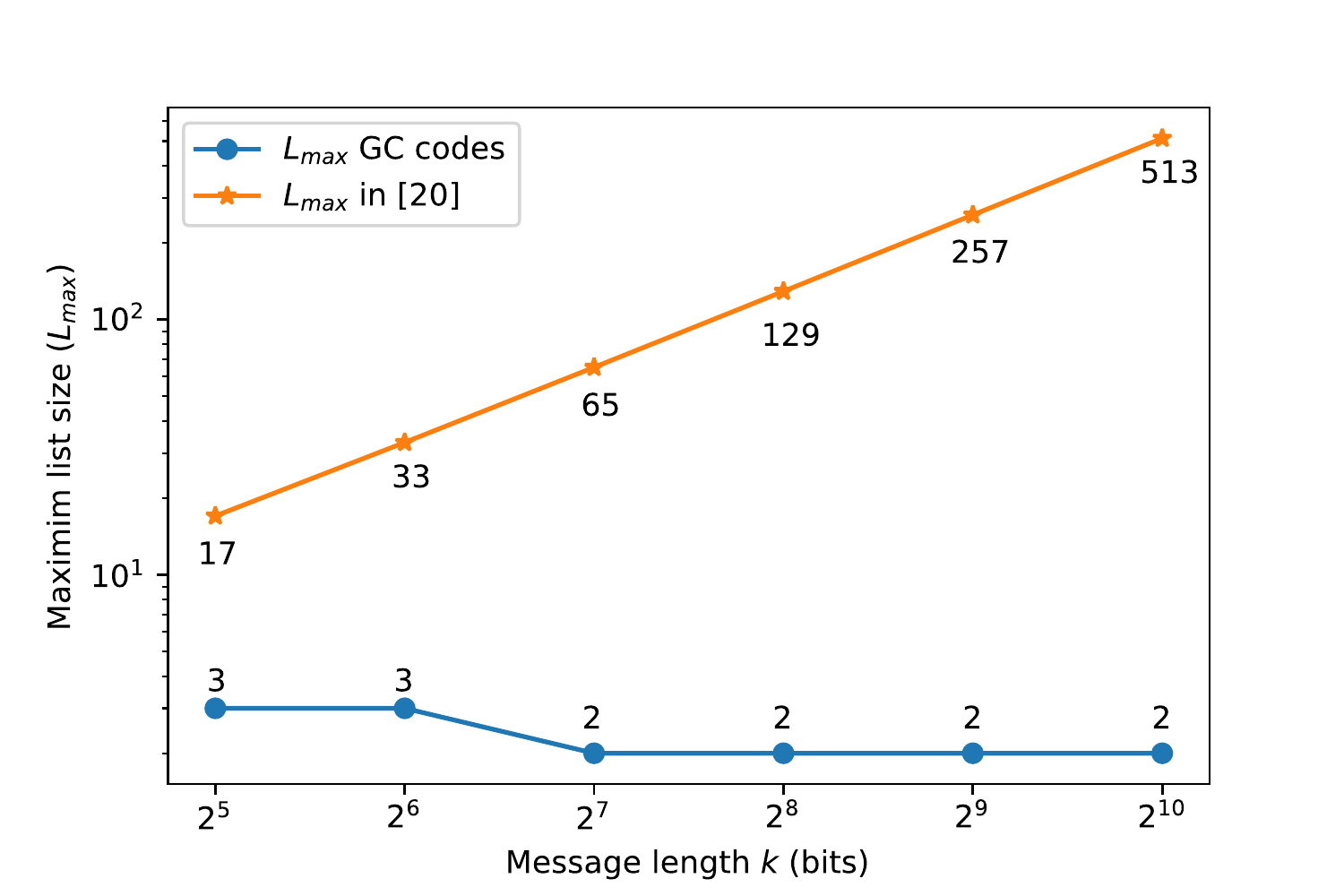}
\caption{The figure shows the maximum list size $L_{max}$ obtained by GC codes and the codes in~\cite{A17} for $\delta=2$ deletions and different message lengths. The GC code parameters are set to $c=\delta+1$ and $\ell=\log k$. The results are obtained over 10000 runs of simulations. In each run, the message and the deletion pattern are chosen independently and uniformly at random.}
\label{fig::2}
\end{figure}

\section{Proofs} 
\label{sec:4} 
\subsection{Proof of Theorem~\ref{thm:1}}
For brevity, we use $L$ instead of $L(\mathbf{u})$ throughout the proof. As previously mentioned, if $L\geq 2$, then we say that the decoder failed to decode uniquely. The probability of failure for a uniform iid binary message $\mathbf{u}$ of length $k$ bits, and any deletion pattern $\mathbf{d}=(d_1,\ldots,d_{\delta})$, is upper bounded by the expression given in~\eqref{eqprf}. Recall that $L\in \{1, \ldots, t\}$, where $t$ is the total number of cases checked by the GC decoder, given by~\eqref{eqt}. The lower bound on $L_{av}$ follows directly from the fact that $L\geq 1$, and hence 
\begin{equation}
\label{eq3}
L_{av}\geq 1.
\end{equation}
To upper bound $L_{aav}$, we write the following.
\begin{align}
L_{av} &= \sum_{l=1}^{\infty} Pr(L\geq l) \label{eqq7}\\
&= \sum_{l=1}^{t} Pr(L\geq l) \\
&= 1 + \sum_{l=2}^{t}  Pr(L \geq l) \label{eq5}\\
&\leq 1 + t Pr(L\geq 2) \label{eq6} \\
&= 1 + \mathcal{O} \left( \frac{k^{\delta}}{\ell^{\delta}}\cdot \frac{k^{\delta}}{\ell^{\delta}2^{\ell(c-\delta)}} \right) \label{eq8} \\
&= 1 + \mathcal{O} \left( \frac{k^{2\delta}}{\ell^{2\delta}2^{\ell(c-\delta)}} \right).
\end{align}
Equations \eqref{eqq7} to \eqref{eq5} follow since $L\in \{1, \ldots, t\}$ is a positive integer-valued random variable. \eqref{eq8}~follows from the fact that $Pr(L\geq 2)=Pr(F)$, and from~\eqref{eqt} and~\eqref{eqprf}.
\subsection{Proof of Corollary 1}
Let $\epsilon>0$, such that $\epsilon=\mathcal{O} \left( k^{2\delta}/\ell^{2\delta}2^{\ell(c-\delta)} \right)$. Based on Theorem~\ref{thm:1} we have 
\begin{equation}
\label{eq11}
1 \leq L_{av}\leq 1+\epsilon.
\end{equation}
To prove that $\lim_{k\to +\infty}L_{av}=1$, we derive conditions on the code parameters under which $\epsilon$ is guaranteed to vanish asymptotically in $k$. $\epsilon$ goes to zero as $k$ approaches infinity if the denominator in its mathematical expression converges to zero faster than the numerator. It is easy to verify that this holds when $\ell=\Omega(\log k)$. Let \mbox{$\ell=\Omega(\log k)$} be the first condition. Then, we have
\begin{equation}
\epsilon= \mathcal{O} \left( \frac{1}{k^{c-3\delta}(\log k)^{2\delta}} \right).
\end{equation}
Hence, we obtain a second condition that $c\geq 3\delta$. Therefore, for $\ell=\Omega(\log k)$ and $c\geq 3\delta$, $L_{av}$ approaches $1$ as the message length $k$ (or equivalently the block length $n$) goes to infinity.

\subsection{Sketch of the Proof of Theorem~\ref{thm:2}}
In what follows, we provide a sketch of the proof of Theorem~\ref{thm:2}. In order to express the variation of the list size with respect to the message length $k$, we use $L(k)$ instead of $L(\mathbf{u})$. $Pr(L(k)=l)$ refers to the probability that the list size is $l$ for uniform iid message of length $k$. Figure~\ref{fig22} illustrates the high-level idea of the proof of Theorem~\ref{thm:2}. The proof is based on the following two lemmas. The proofs of these lemmas are given in the appendix.
\begin{lemma}[informal]
For any fixed list size greater than one, there exists an infinitely increasing sequence of message lengths over which $Pr(L(k)=l)$ is non-increasing.
\end{lemma}
\begin{lemma} [informal]
For any fixed message length, if $Pr(L(k)=l)=0$, then $Pr(L(k)>l)=0$.
\end{lemma}
The high-level idea of the proof is the following. Consider a sufficiently large fixed message length $k_1$. By definition, we have $Pr\left(L(k)\geq L_{max}(k_1)+1\right)=0$, where $L_{max}(k_1)$ is the maximum list size of any message of length $k_1$. Based on Lemma~\ref{lemma1}, for the fixed list size $L_{max}(k_1)+1$, there exists an infinitely increasing sequence of message lengths $\{k_1,k_2,k_3,\ldots\}$, such that $Pr(L(k)=l)$ is non-increasing over this set of message lengths. Therefore, $\forall k\in \{k_1,k_2,\ldots\}, Pr(L(k)=L_{max}+1)=0$. Furthermore, based on Lemma~\ref{lemma2}, $\forall k\in \{k_1,k_2,\ldots\}, Pr(L(k)>L_{max}+1)=0$. By combining these results, $\forall k\in \{k_1,k_2,\ldots\}$ we have $Pr\left(L(k)> L_{max}(k_1)\right)=0$, and hence the maximum list size is upper bounded by the constant $L_{max}(k_1)$.

\begin{figure}[h!]
\centering
\begin{tikzpicture}
\node (a) {
\begin{tabular}[h!]{|c|cccccccc}
\\ \hline
\backslashbox{$k$}{$l$} & $1$ & $\cdots$ & $L_{max}(k_1)$ & $L_{max}(k_1)+1$ & & $L_{max}(k_1)+2$ & $\cdots$ & $\cdots$ \\ \hline
$\vdots$ \vspace{0.15cm}&  & & & \\ 
$k_1$ & $>0$ & $\cdots$ & $>0$ & $\boxed{\mathbf{0}}$ & & $\mathbf{0}$ & $\mathbf{0}$ & $\cdots$\\
$\vdots$ \vspace{-0.1cm}& &  & & \hspace{0.03cm} $\Big\downarrow$ \vspace{0.2cm} &  & & \\
$k_2$ & $\cdots$ & $\cdots$ & $\cdots$ & $\mathbf{0}$ & \hspace{-1.6cm} $\xrightarrow[]{\text{Lemma 2}}$ \hspace{-1.6cm} & $\mathbf{0}$ & $\mathbf{0}$ & $\cdots$ \\
$\vdots$ \vspace{-0.1cm}& &  & & \hspace{0.03cm} $\Big\downarrow$ \vspace{0.2cm} &  & & \\
$k_3$ & $\cdots$ & $\cdots$ & $\cdots$ & $\mathbf{0}$ & \hspace{-1.6cm} $\xrightarrow[]{\text{Lemma 2}}$ \hspace{-1.6cm} & $\mathbf{0}$ & $\mathbf{0}$ & $\cdots$ \\
$\vdots$ \vspace{0.15cm}& & & &\\
\end{tabular}};
\node (b) at (-0.55,0.15) {\scriptsize Lemma 1};
\node (c) at (-0.55,-1.35) {\scriptsize Lemma 1};

\path[->, line width = 0.15mm] (-6.6,1.8) edge node[xshift = -0.6cm, font = \scriptsize] {$\substack{\mbox{Message}   \\ \mbox{length}}$} (-6.6,0);

\path[->, line width = 0.15mm] (-5,3) edge node[yshift = 0.2cm, font = \scriptsize] {List size} (-3,3);
\end{tikzpicture}
\caption{$Pr(L(k)=l)$ for different message lengths $k$ and list sizes $l$. For the fixed list size $L_{max}(k_1)+1$, Lemma~\ref{lemma1} implies that $Pr(L(k)=l)$ is non-increasing for $k\in \{k_1,k_2,k_3,\ldots\}$ (vertical direction). For a fixed message size, Lemma~\ref{lemma2} implies that $Pr(L(k)>l)=0$ if $Pr(L(k)=l)=0$ (horizontal direction).}
\label{fig22}
\end{figure}
\vspace{-0.5cm}
\subsection{Proof of Theorem~\ref{thm:2}}
As previously mentioned, for a uniform iid message of length $k$ bits, $Pr(L(k)\geq 2)=Pr(F)$, where $Pr(F)$ is the probability that the decoder fails to decode uniquely. Therefore, from~\eqref{eqprf}, for any $l\geq 2$ we have
\begin{equation}
Pr(L(k)=l)=\mathcal{O}\left( f(k) \right),
\end{equation}
where $f(k)\triangleq k^{\delta}/\ell^{\delta}2^{\ell(c-\delta)}$. It can be easily verified that \mbox{$f(k)>0$} is a strictly decreasing function of~$k$ if \mbox{$\ell=\Omega(\log k)$} and $c\geq 2\delta$.
\addtocounter{lemma}{-2}
\begin{lemma}
\label{lemma1}
For any fixed $l\in\{2,\ldots,t\}$, $\forall k\gg 0$, $\exists k'>k$, such that
\begin{equation*}
Pr(L(k')=l)  \leq Pr(L(k)=l).
\end{equation*}   
\end{lemma}
\begin{proof}
See Appendix A.
\end{proof}
\begin{lemma}
\label{lemma2}
For any fixed $k$, $\forall l\in\{2,\ldots,t\}$, if $Pr(L(k)=l)=0$, then $Pr(L(k)>l)=0$.
\end{lemma}
\begin{proof}
See Appendix B.
\end{proof}
\noindent Consider a sufficiently large fixed message length $k_1>0$. Lemma~\ref{lemma1} implies that, for any fixed $l\geq 2$, there exists an infinitely increasing sequence \mbox{$(k_1,k_2,\ldots)$}, such that $Pr(L(k)=l)$ is decreasing in $k$, for $k\in S\triangleq \{k_1, k_2, \ldots \}$. Let $S$ be the sequence associated with the fixed list size \mbox{$l^{\star} \triangleq L_{max}(k_1)+1$}, where $L_{max}(k_1)$ is the maximum decoding list size for any message of length~$k_1$~bits. By definition, we have
\begin{equation}
\label{eq15}
Pr\left(L(k_1)\geq l^{\star}\right)=0.
\end{equation}
Consider the set $S-\{k_1\}$, $\forall k\in S-\{k_1\}$ we have
\begin{align}
Pr\left(L(k)= l^{\star}\right) &\leq Pr\left(L(k_1)= l^{\star}\right) \label{eqq1}\\
&= 0. \label{eqq2}
\end{align}
\eqref{eqq1} follows from the fact that $Pr(L(k)=l^{\star})$ is non-increasing for $k\in S$. \eqref{eqq2} follows from~\eqref{eq15}. From~\eqref{eq15}~and~\eqref{eqq2}, we can deduce that $\forall k\in S$,
\begin{equation}
\label{eqq5}
Pr\left(L(k)= l^{\star}\right)=0.
\end{equation}
Furthermore, from Lemma~\ref{lemma2} we have, $\forall k\in S$,
\begin{align}
Pr\left(L(k)> l^{\star}\right) =0. \label{eqq4}
\end{align}
By combining the results from~\eqref{eqq5}~and~\eqref{eqq4} we get that \mbox{$\forall k\in S$},
\begin{equation}
\label{eqq6}
Pr\left(L(k)\geq l^{\star}\right)=0.
\end{equation}
Hence, $\forall k\in S$,
\begin{equation}
\label{eq17}
L_{max}(k)<l^{\star}.
\end{equation}
This proves that there exists an infinite sequence of GC codes, indexed by $S$, whose maximum list size $L_{max}$ is upper bounded by a constant\footnote{$l^{\star}=L_{max}(k_1)+1$ is a constant that does not increase with $k$.}. 
\begin{remark}
The result of Theorem~\ref{thm:2} can be used to improve the upper bound in Theorem~\ref{thm:1}. Namely, based on Theorem~\ref{thm:2}, in~\eqref{eq6} we can upper bound the list size by a constant $\mathcal{O}(1)$ instead of the total number of guesses $t$ given by~\eqref{eqt}. However, in that case, the result on $L_{av}$ would be restricted to a sequence of GC codes, whereas Theorem~\ref{thm:1} is a universal result for all GC codes.
\end{remark}

\appendices
\setcounter{lemma}{0}
\vspace{-0.5cm}
\section{Proof of Lemma~\ref{lemma1}} \label{app:A}
\begin{lemma}
For any fixed $l\in\{2,\ldots,t\}$, $\forall k\gg 0$, $\exists k'>k$, such that
\begin{equation*}
Pr(L(k')=l)  \leq Pr(L(k)=l).
\end{equation*}   
\end{lemma}
\begin{proof}
\noindent Lemma~\ref{lemma1}~follows from the fact that for any fixed $l\geq 2$,
\begin{equation*}
Pr(L(k)=l)=\mathcal{O}\left( f(k) \right),
\end{equation*}
where $f(k)=k^{\delta}/\ell^{\delta}2^{\ell(c-\delta)}$ is a strictly decreasing function of~$k$ for $\ell=\Omega(\log k)$ and $c\geq 2\delta$. Therefore, for any fixed $l\geq 2$,
\begin{equation*}
\lim_{k\to +\infty} Pr(L(k)=l)=0.
\end{equation*}
Hence, by the definition of limits, $\forall k\gg 0$, $\exists k'>k$, such that
\begin{equation*}
Pr(L(k')=l)  \leq Pr(L(k)=l).
\end{equation*} 
\end{proof}

\vspace{-0.5cm}
\section{Proof of Lemma~\ref{lemma2}}
\label{app:B}
\begin{lemma}
For any fixed $k$, $\forall l\in\{2,\ldots,t\}$, if $Pr(L(k)=l)=0$, then $Pr(L(k)>l)=0$.
\end{lemma}
\begin{proof}
Lemma~\ref{lemma2}~follows from the construction of GC codes described in Section~\ref{sec:2}. To prove Lemma~\ref{lemma2}, we first introduce some notations. Consider a message $\mathbf{u}$ of fixed length~$k$ and an arbitrarily fixed deletion pattern \mbox{$\mathbf{d}=(d_1,\ldots,d_{\delta})$}. Let \mbox{$\mathtt{Enc}: \mathbb{F}_2^k \rightarrow \mathbb{F}_2^n$} be the GC encoding function~(Fig.~\ref{fig:1}) that maps the message to the corresponding codeword. Let \mbox{$\mathtt{Del}: \mathbb{F}_2^n\rightarrow \mathbb{F}_2^{n-\delta}$} be the deletion function that maps the transmitted codeword to the received string based on the deletion pattern \mbox{$\mathbf{d}$}. Recall that the GC decoder generates~$t$~guesses where each guess corresponds to a string that is decoded based on a certain assumption on the locations of the $\delta$ deletions. Let \mbox{$\mathtt{Dec}_i: \mathbb{F}_2^{n-\delta} \rightarrow \mathbb{F}_2^n$}, be the function that decodes the received string based on the~$i^{th}$~guess, where~\mbox{$i=1,\ldots,t$}. 

Since the encoding and decoding algorithms of GC codes are deterministic; and the deletion pattern $\mathbf{d}$ is fixed, then the functions $\mathtt{Enc},\mathtt{Del},$ and $\mathtt{Dec}$ are deterministic functions of their inputs. Let $\mathtt{g}_i:\mathbb{F}_2^k \rightarrow \mathbb{F}_2^n$ be the composition of these functions, given by \mbox{$\mathtt{g}_i\triangleq \mathtt{Dec}_i \circ \mathtt{Del} \circ \mathtt{Enc}$}, for~\mbox{$i=1,\ldots,t$}. Note that $\mathcal{G}=\{\mathtt{g}_i : i=1,\ldots,t\}$ is a family of distinct deterministic functions from $\mathbb{F}_2^k$ to $\mathbb{F}_2^n$ . 

For a given message $\mathbf{u}$, let $\{i_1,\ldots,i_m\}$ be the set of indices of the functions $\mathtt{g}_i(\mathbf{u})$ that result in different decoded strings. Namely, $\forall i,j\in \{i_1,\ldots,i_m\}$, such that $i\neq j$, we have $\mathtt{g}_i(\mathbf{u})\neq \mathtt{g}_j(\mathbf{u})$. The guessing phase results in an initial list of~$m\leq t$ decoded strings given by
\begin{equation*}
\mathcal{L}_{initial}(\mathbf{u})=\{\mathtt{g}_{i_1}(\mathbf{u}),\mathtt{g}_{i_2}(\mathbf{u}),\ldots,\mathtt{g}_{i_m}(\mathbf{u})\}.
\end{equation*}

In the checking phase, each string in \mbox{$\mathcal{L}_{initial}(\mathbf{u})$} is checked with $c-\delta$ MDS parities in order to eliminate the strings which are not consistent with these parities. Let $\mathbf{A}\in \mathbb{F}_q^{n\times (c-\delta)}$ be the matrix consisting of the encoding vectors corresponding to the $c-\delta$ MDS parities. Let $\mathbf{b}(\mathbf{u})\in \mathbb{F}_q^{c-\delta}$ be the vector containing the values of the $c-\delta$ parities corresponding to the encoding of~$\mathbf{u}$. The final list of candidate strings given by
\begin{align*}
\mathcal{L}_{final}(\mathbf{u})&=\{\mathbf{x}\in \mathcal{L}_{initial}(\mathbf{u}) : \mathbf{A}^T\mathbf{x}\equiv \mathbf{b}(\mathbf{u})\Mod{q}\}.
\end{align*}
For brevity, henceforth we refer to $\mathcal{L}_{final}(\mathbf{u})$ by $\mathcal{L}$. Let $l\in\{1,\ldots,m\}$, consider the following sets 
\begin{align*}
B&\triangleq \{\mathbf{u}\in \mathbb{F}_2^k : \mathtt{g}_{i_{1}}(\mathbf{u})\in \mathcal{L}, \ldots, \mathtt{g}_{i_{l-1}}(\mathbf{u})\in \mathcal{L}, \mathtt{g}_{i_{l+1}}(\mathbf{u})\notin \mathcal{L}, \ldots, \mathtt{g}_{i_{m}}(\mathbf{u})\notin \mathcal{L} \}, \\
C&\triangleq \{\mathbf{u}\in \mathbb{F}_2^k : \mathtt{g}_{i_{l}}(\mathbf{u})\in \mathcal{L} \}, \\
\overline{C}&\triangleq \{\mathbf{u}\in \mathbb{F}_2^k : \mathtt{g}_{i_{l}}(\mathbf{u})\notin \mathcal{L} \}. 
\end{align*}
\noindent Note that the condition $\mathtt{g}_{i_{l}}(\mathbf{u})\in \mathcal{L}$ in set $C$ is equivalent to
\begin{equation}
\label{set1}
\mathbf{A}^T\mathtt{g}_{i_{l}}(\mathbf{u})\equiv \mathbf{b}(\mathbf{u})\Mod{q}. 
\end{equation}
The next claim states that $\left \lvert B \cap C\right \rvert \leq \left \lvert B \cap \overline{C} \right \rvert$, which is intuitive since the intersection with $C$ introduces an additional equality constraint which is generally more restrictive than the ``not equal" constraint of $\overline{C}$.
\begin{claim}
\label{claim1}
For any $l\in\{1,\ldots,m\}, \left \lvert B \cap C\right \rvert \leq \left \lvert B \cap \overline{C} \right \rvert.$
\end{claim}
\begin{proof}
Since $|B|=\left \lvert B \cap C\right \rvert+\left \lvert B \cap \overline{C}\right \rvert$, then it is easy to see that proving $\left \lvert B \cap C\right \rvert \leq \left \lvert B \cap \overline{C} \right \rvert$ is equivalent to proving that \mbox{$\left \lvert B \cap C\right \rvert \leq |B|/2$}. If $B$ is empty, then it is trivial that $\left \lvert B \cap C\right \rvert \leq |B|/2$. If $B$ is non-empty, then we consider the following cases:
\begin{itemize}

\item \underline{$B\cap C=\emptyset$:} Then, $\left \lvert B \cap C\right \rvert = 0 < |B|/2$, and hence the inequality is satisfied.

\item \underline{$B\cap C=B$:} Since $\mathbf{A}$ and $\mathbf{b}(\mathbf{u})$ are fixed for all functions $\mathtt{g}_{i}(\mathbf{u}), i=i_1,\ldots,i_m$, then $B\cap C=B \Rightarrow \exists j\in \{1,\ldots,l-1\}$, such that
\begin{equation}
\label{set2}
\forall \mathbf{u}\in \mathbb{F}_2^n,~\mathbf{A}^T\mathtt{g}_{i_{l}}(\mathbf{u})\equiv \mathbf{b}(\mathbf{u})\Mod{q} \Leftrightarrow \mathbf{A}^T\mathtt{g}_{i_j}(\mathbf{u})\equiv \mathbf{b}(\mathbf{u})\Mod{q}.
\end{equation}
Since $q=2^{\ell}$ is a prime power, then~\eqref{set2} implies that
\begin{equation*} 
\forall \mathbf{u}\in \mathbb{F}_2^n,~\mathtt{g}_{i_{l}}(\mathbf{u})= \mathtt{g}_{i_j}(\mathbf{u}),
\end{equation*}
with $j\neq l$, which contradicts the fact that $\mathcal{G}=\{\mathtt{g}_i : i=i_1,\ldots,i_m\}$ is a family of distinct functions. Therefore, the case of $B\cap C=B$ is not possible.

\item \underline{$B\cap C\neq \emptyset$ and $B\cap C\neq B$:}  The intersection of set $C\subset \mathbb{F}_2^{k}$ with set $B\subset \mathbb{F}_2^{k}$ can be seen as introducing an additional equality constraint to the set $B$, given by~\eqref{set1}. If $B\cap C\neq \emptyset$ and $B\cap C\neq B$, then~\eqref{set1} is satisfiable by some of the messages $\mathbf{u}\in B$. The constraint given by~\eqref{set1} corresponds to a set of linear congruences in $\mathbb{F}_q$, with $q=2^{\ell}>2$, where the variables in these congruences are the bits of the message $\mathbf{u}$. Since the messages $\mathbf{u}\in B\cap C$ must satisfy these linear congruences, then the degree of freedom of the bits of $\mathbf{u}$ is decreased by at least $1$. In other words, the size of the set $B\subset \mathbb{F}_2^{k}$ is decreased by at least a factor of $2$, i.e., $\left \lvert B \cap C\right \rvert \leq |B|/2$.
\end{itemize}
\end{proof}
Recall that $L(k)=L(\mathbf{u})=\left \lvert \mathcal{L}(\mathbf{u})\right \rvert$ is the final list size defined in~\eqref{def1}, and $Pr(L(k)=l)$ refers to the probability that the list size is exactly $l$ for a uniform iid message $\mathbf{u}$. Next, we prove that if $Pr(L(k)=l-1)=0$, then $Pr(L(k)>l-1)=0$. Note that proving this statement is equivalent to proving Lemma~\ref{lemma2}.

The event $\{L(k)=l\}$ can result from $\binom{m}{l}$ cases, where each case corresponds a certain combination of $l$ out of the $m$ functions $\mathtt{g}_{i}(\mathbf{u}), i=i_1,\ldots,i_m$, satisfying equality constraints of the form of~\eqref{set1}, whereas the remaining $m-l$ functions do not satisfy these constraints. The set $B \cap C$ represents one of these $\binom{m}{l}$ cases which lead to $\{L(k)=l\}$. We index these cases by $j=1,\ldots,\binom{m}{l}$, and refer to each case by its corresponding set $B_j \cap C_j$. The sets $B_j \cap C_j, j=1,\ldots,\binom{m}{l}$, are mutually disjoint because the conditions on $\mathbf{u}$ that define these sets are contradictory. Furthermore, it is easy to see that the result of Claim~1 applies to all of the $\binom{m}{l}$ cases, i.e., $\forall j\in \{1,\ldots,\binom{m}{l}\}$, we have $\left \lvert B_j \cap C_j \right \rvert\leq \left \lvert B_j \cap \overline{C_j} \right \rvert$. Therefore,
\begin{align}
Pr(L(k)=l) &= Pr\left(\bigcup_{j=1}^{\binom{m}{l}} \{ \mathbf{u}\in B_j \cap C_j \}\right) \label{seteq1} \\
&= \sum_{j=1}^{\binom{m}{l}} Pr\left( \mathbf{u}\in B_j \cap C_j \right) \label{seteq2} \\
&= \sum_{j=1}^{\binom{m}{l}} \frac{\left \lvert B_j \cap C_j \right \rvert}{2^k} \label{seteq3}\\
&\leq \sum_{j=1}^{\binom{m}{l}} \frac{\left \lvert B_j \cap \overline{C_j} \right \rvert}{2^k} \label{seteq4}. 
\end{align}
\eqref{seteq2} follows from the fact that the $\binom{m}{l}$ cases correspond to mutually disjoint sets. \eqref{seteq3} follows from the fact the message is uniform iid. \eqref{seteq4} follows from Claim 1. Note that $\forall \mathbf{u}\in B_j \cap \overline{C_j}$, we have $\left \lvert \mathcal{L}(\mathbf{u}) \right \rvert = l-1$. Hence, if $Pr(L(k)=l-1)=0$, then $\forall j\in \{1,\ldots,\binom{m}{l}\}$, $\left \lvert B_j \cap \overline{C_j} \right \rvert = 0$, and hence from~\eqref{seteq4} we get $Pr(L(k)=l)\leq 0$. This proves that if $Pr(L(k)=l-1)=0$, then $Pr(L(k)=l)=0$. Similarly, $\forall l'>l$, we can apply the same approach to get $Pr(L(k)=l')=0$. Therefore, if $Pr(L(k)=l-1)=0$, then $Pr(L(k)>l-1)=0$.
\end{proof}

\bibliographystyle{ieeetr}

\bibliography{Refs}

\end{document}